\documentclass{article}
\usepackage{ijcai24}
\usepackage{times}
\usepackage{soul}
\usepackage{url}
\usepackage[hidelinks]{hyperref}
\usepackage[utf8]{inputenc}
\usepackage[small]{caption}
\usepackage{graphicx}
\usepackage{amsmath}
\usepackage{amsthm}
\usepackage{booktabs}
\usepackage[switch]{lineno}

\newtheorem{theorem}{Theorem}
\newtheorem{lemma}{Lemma}
\newtheorem{corollary}{Corollary}
\newtheorem{definition}{Definition}
\newtheorem{hypothesis}{Hypothesis}
\newtheorem{conjecture}{Conjecture}

\author{
    Ben Bals$^{1,2}$
    \and
    Michelle Döring$^{3}$\and
    Nicolas Klodt$^{3}$\And
    George Skretas$^3$\\
    \affiliations
    $^1$Centrum Wiskunde \& Informatica, Amsterdam\\
    $^2$Vrije Universiteit Amsterdam\\
    $^3$Hasso Plattner Institute, Potsdam\\
    \emails
    bjjb@cwi.nl,
    \{michelle.doering, nicolas.klodt, georgios.skretas\}@hpi.de,
}

\title{Dynamic Network Discovery via Infection Tracing} 

\newcommand{\mm}[1]{\textcolor{orange}{Mi: #1}}

\usepackage[disable]{todonotes}
\usepackage{placeins}
\usepackage[normalem]{ulem}
\usepackage[adversary]{cryptocode}
\usepackage{bm}
\usepackage{tabularx}
\usepackage{makecell}
\usepackage{subcaption}
\usepackage{amsmath}
\usepackage{amssymb}
\usepackage{enumitem}
\usepackage{booktabs}
\usepackage{appendix}
\usepackage[appendix=append]{apxproof}
\usepackage{layouts}

\usepackage[ruled,noend]{algorithm2e}

\ifdefined\N\else\DeclareMathOperator{\N}{\mathbb{N}}\fi
\ifdefined\R\else\DeclareMathOperator{\R}{\mathbb{R}}\fi

\newcommand{\BigO}{\mathcal{O}}
\newcommand{\iphase}{\delta}
\newcommand{\nodeset}[1]{V(#1)}
\newcommand{\edgeset}[1]{E(#1)}

\newcommand{\setsize}[1]{\left|#1\right|}
\newcommand{\nodesetsize}[1]{\setsize{V(#1)}}
\newcommand{\edgesetsize}[1]{\setsize{E(#1)}}
\newcommand{\decc}[1]{\setsize{EC_\iphase(#1)}}
\newcommand{\Tmax}{T_\mathrm{max}}
\DeclareMathOperator{\edgeLabelOp}{\lambda}
\newcommand{\edgeLabel}[1]{\edgeLabelOp(#1)}
\newcommand{\G}{\mathcal{G}}
\newcommand{\defaultGraphName}{\G}

\newcommand{\defaultdecc}{\decc{\defaultGraphName}}

\newcommand{\FMainHack}{\mathtt{Follow}}
\newcommand{\FMain}{\FMainHack}
\newcommand{\FDMainHack}{\mathtt{DiscoveryFollow}}
\newcommand{\FDMain}{\FDMainHack}
\newcommand{\exHack}{\mathtt{Explore}}
\newcommand{\ex}{\exHack}

\newif\ifpaper
\definecolor[named]{benGray}{rgb}{0.61,0.61,0.63}
\newcommand{\dECC}{\(\iphase\)-ecc}
\newcommand{\dECCs}{\(\iphase\)-ecc's}
\newcommand{\ipz}{IPZ}

\setlist[itemize]{noitemsep, nolistsep}

\papertrue

\newtheoremrep{appendixtheorem}[theorem]{$\star$ Theorem}
\newtheoremrep{appendixlemma}[lemma]{$\star$ Lemma}
\newtheoremrep{appendixcorollary}[corollary]{$\star$ Corollary}

\usepackage[capitalise]{cleveref}
\crefname{hypothesis}{Hypothesis}{Hypotheses}
\crefname{conjecture}{Conjecture}{Conjectures}

\begin{document}

\maketitle

\begin{abstract}
Researchers, policy makers, and engineers need to make sense of data from spreading processes as diverse as 
rumor spreading in social networks, viral infections, and water contamination.
Classical questions include predicting infection behavior in a given network or deducing the network structure from infection data.
Most of the research on network infections studies \emph{static graphs}, that is, the connections in the network are assumed to not change. 
More recently, \emph{temporal graphs}, in which connections change over time, have been used to more accurately represent real-world infections, which rarely occur in unchanging networks.
We propose a model for temporal graph discovery that is consistent with previous work on static graphs and embraces the greater expressiveness of temporal graphs.
For this model, we give algorithms and lower bounds which are often tight. We analyze different variations of the problem, which make our results widely applicable and it also clarifies which aspects of temporal infections make graph discovery easier or harder. 
We round off our analysis with an experimental evaluation of our algorithm on real-world interaction data from the Stanford Network Analysis Project and on temporal Erdős-Renyi graphs.
On Erdős-Renyi graphs, we uncover a threshold behavior, which can be explained by a novel connectivity parameter that we introduce during our theoretical analysis.\footnotemark%
\footnotetext{This work is based on the first author's master thesis, which is available at \url{https://arxiv.org/abs/2503.13567}.}
\end{abstract}

\section{Introduction}
\label{sec:org1e46807}

Predicting the spread of infections requires precise knowledge about the network in which they take place.
These spreading processes can be vastly different; they involve everything from diseases, misinformation, marketing material to contaminants in sewage networks. 
All of these can be modeled in a similar fashion, and we can thus utilize a common algorithmic toolkit for their analysis.
Most famously, the \emph{influence maximization problem}, introduced by
\cite{independent-cascade},
wants to find which node in a network should be infected to maximize the number of nodes infected by the spreading process.
Other problems include finding a sensor placement to detect outbreaks as quickly as possible \cite{sensor-placement}, or to estimate the source (of an infection or rumor) from data about the spreading process \cite{pmlr-v216-berenbrink23a,infection-sources-geometric,sir-source-detection}.


One common assumption is that the underlying network is known.
However, this is not true in many real-world scenarios, and thus, we require algorithms for \emph{network discovery}. 
While discovering networks is a fundamental problem in data mining, which can be approached from different angles \cite{estimate-diffusion-networks,park2016information}, one natural approach is to discover the underlying network from the infection data itself.
This idea has been extensively studied \cite{infer-from-cascades,lokhov2016reconstructing,daneshmand2014estimating,netrapalli2012learning}.
In particular, \cite{chistikov2024learning} consider a model where the party wishing to discover the network may even intervene in the spreading process, e.g., by publishing a social media post and watching its spread.
Beyond this application, 
network discovery is an interesting and relevant problem in its own right.
After having discovered the network from infection data, we are free to abstract away from the spreading process and use the resulting network in a host of different ways.
This is especially relevant for the study of social networks, both real-world and online, where infection data can reveal underlying structures that are otherwise difficult to observe.

To the best of our knowledge, every paper that studies network discovery makes the same simplifying assumption: the underlying network is \emph{static}.
That is, the connections of the network do not change over time.
In most applications, that is not a realistic assumption.
For example, if two people are linked in an in-person social network, that does not imply that a disease can spread from one to the other at every point in time, but only when they physically meet.
Motivated by this fact, researchers have begun to study the classical infection analysis tasks on \emph{temporal graphs}.

Temporal graphs are a model of dynamic networks where the edges only exist at some time steps.
This model has received considerable attention from theoretical computer scientists for both foundational problems \cite{michail2016introduction,Danda,Casteigts2021FindingTP,akrida2019temporal}and a growing number of applications, including social networks \cite{casteigts_et_al:DagRep.11.3.16}.
For our purposes, a temporal graph \(\G=(V, E, \edgeLabelOp)\) with lifetime \(\Tmax \in \N\) is a static graph \(G=(V, E)\) with a function \(\edgeLabelOp\colon E \to 2^{[\Tmax]}\) indicating that edge \(e \in E\) exists precisely at the time steps \(\edgeLabel{e}\).
\cite{gayraud-evolving-social-networks},
are the first to study the influence maximization problem on temporal graphs under the independent cascade model (introduced by \cite{independent-cascade}).
\cite{influencers} build on this work and analyze the influence maximization problem on temporal graphs under the SIR model (a standard biological spreading model closely related to the susceptible-infected-resistant model \cite{Hethcote1989}).
However, no work on network discovery on temporal graphs has been conducted yet.
\paragraph*{Our Contribution}\label{contribution}is twofold: (i) we define the temporal network discovery problem as a round-based, interactive two-player game, (ii) we provide algorithms and lower bounds for different parameters of the network discovery game, and validate our results both with theoretical proofs, as well as experiments conducted on real-world and synthetic networks. 
\noindent{\emph{Statements where proofs or details are omitted to the appendix due to
space constraints are marked with $\star$.}}

In \Cref{sec:game}, we define the two-player game. In each round, the \emph{Discoverer} (abbreviated \(\ddv\)) initiates infections and observes the resulting infection chains, aiming to identify the time labels of all edges in as few rounds as possible.
The \emph{Adversary} (abbreviated \(\adv\)) gets to pick the shape and temporal properties of the graph, with the aim of forcing \(\ddv\) to take as long as possible to accomplish their task.

In \Cref{sec:ideal-patient-zero}, we provide the \(\FDMain\) algorithm, which solves the graph discovery problem in \(6\edgesetsize{\G} + \decc{\G} \left\lceil \Tmax/\iphase \right \rceil\) rounds, where \(\decc{\G}\) are the \(\iphase\)-edge connected components. Intuitively, this is a grouping of the edges such that only edges from the same component may influence each other during infection chains.
In \Cref{sec:witnesses}, we prove that the running time of the \(\FDMain\) algorithm is asymptotically tight in the number of edges.
Formally, we prove there is an infinite family of graphs such that any algorithm winning the graph discovery game requires at least \(\Omega(\edgesetsize{\G})\) rounds.
Crucially, this cannot be improved even if \(\ddv\) is allowed to start multiple infection chains per round.
We also prove that there is an infinite family of graphs such that the minimum number of rounds required to win the graph discovery game grows in \(\Omega(n T_\mathrm{max} / (\iphase k))\), where \(k\) is the number of infection chains \(\ddv\) may start per round.

\newcommand{\graphDiscoveryResultsTable}{%
\begin{table}[t]
\caption{\label{tbl:variations-overview}Overview of upper and lower bounds on the number of rounds for different variations of the graph discovery game.
A subscript to a Landau symbol indicates variables that the asymptotic growth is independent of.\todo{is there a reason that the caption is on top?}}
\centering
\begin{tabular}{ll}
\toprule
\textbf{Lower Bound} & \textbf{Upper Bound}\\[0pt]
\midrule
\textbf{Basic model} & \\ \(\Omega_k(m)\) & \(\BigO_k(m + \defaultdecc{} \Tmax / \iphase)\)\\[0pt]
\midrule
\textbf{Infection times only} & \\ \(\Omega_k(m)\) & \(\BigO_k(m + \defaultdecc{} \Tmax / \iphase)\)\\[0pt]
\midrule
\textbf{Unknown static graph} & \\ \(\Omega_{m, \decc{\G}}(n \Tmax / (\iphase k))\) & \(\BigO(n \Tmax )\)\\[0pt]
\midrule
\textbf{Multilabels} & \\ \(\Omega_{\decc{\G}}(n \Tmax / (\iphase k))\) & \(\BigO(n \Tmax )\)\\[0pt]
\midrule
\textbf{Multiedges} & \\ \(\Omega(n \Tmax / (\iphase k))\) & \(\BigO_k(m + \defaultdecc{} \Tmax / \iphase)\)\\[0pt]
\bottomrule
\end{tabular}
\end{table}
}
 \ifpaper \else {\graphDiscoveryResultsTable}\fi

We finish our theoretical analysis in \Cref{sec:extending}, where we explore variations of the graph discovery problem.
We analyze the case where the feedback \(\ddv\) receives about the infection chains is reduced to infection times.
Surprisingly, we are able to show that our \(\FDMain\) algorithm directly translates to this scenario. 
We also discuss what happens if \(\ddv\) has no information about the static graph in which the infections are taking place.
Third, we allow the temporal graph to now contain multiedges or more than one label per edge.

In \Cref{sec:experiments}, we empirically validate our theoretical results.
Using both synthetic and real-world data, we execute the \(\FDMain\) algorithm and observe its performance.
We utilize the natural temporal extension of Erdős-Renyi graphs \cite{casteigts_threshold} as well as the \texttt{comm-f2f-Resistance} data set from the Stanford Large Network Dataset Collection \cite{kumar2021deception}, a social network of face-to-face interactions.
Beyond the running time, we closely analyze which factors affect the performance of the algorithm.
We see that the density of the graph affects the performance since, in dense graphs, it needs to spend less time finding new \(\iphase\)-edge connected components.
On Erdős-Renyi graphs, we provide evidence that this effect is mediated by the number of \(\iphase\)-edge connected components, which exhibits a threshold behavior in \(\Tmax/(\nodesetsize{\G} p)\), where \(p\) is the Erdős-Renyi density parameter.
This prompts us to give a conjecture on this threshold behavior, which mirrors the famous threshold behavior in the connected components of nodes in static Erdős-Renyi graphs \cite{erdos1960evolution}.


\section{Preliminaries} \label{sec:prelims}
\label{prelims}
For \(n \in \N^+\), \(x \in \N\), let \([x, n] \coloneqq \{x, \dots, n\}\) and \([n] \coloneqq [1,n]\).

A \emph{temporal graph} \(\G = (V, E, \edgeLabelOp)\) with lifetime \(\Tmax\) is composed of an undirected \emph{(underlying) static graph} \(G=(V, E)\) together with a \emph{labeling function} \(\edgeLabelOp \colon E \to 2^{[\Tmax]}\), denoting \(e \in E\) being present precisely at time steps \(\edgeLabel{e}\).
We also write \(\nodeset{\G}\) for the nodes of \(\G\), and \(\edgeset{\G}\) for its edges.
A \emph{temporal path} a sequence labeled edge $(e_1,\edgeLabelOp{e_1}),\dots,(e_{\ell},\edgeLabelOp{e_\ell})$ that forms a path in $G$ with strictly increasing labels, i.\,e., for all $i\in[\ell]$,  $\edgeLabelOp{e_{i}}<\edgeLabelOp{e_{i+1}}$.
Apart from \Cref{sec:extending}, we consider \emph{simple} temporal graphs where each edge has exactly one label. Abusing  notation, we therefore also use \(\edgeLabelOp\) as if it were defined as \(\edgeLabelOp\colon E \to [\Tmax]\), and regard temporal paths as the corresponding sequence of nodes.

We use the susceptible-infected-resistant (SIR) model, in which a node is either in a \emph{susceptible}, \emph{infected}, or \emph{resistant} state. 
This model of temporal infection behavior is based on \cite{influencers} (and more historically flows from \cite{a-contribution} and \cite{pastor2015epidemic}).
An \emph{infection chain} in the SIR model unfolds as follows.
At most \(k \in \N\) nodes may be infected by \(\ddv\) at arbitrary points in time, which we call \emph{seed infections} denoted as \(S \subseteq V \times [0,\Tmax]\). 
Otherwise, a node \(u\) becomes \emph{infected at time step \(t\)} if and only if it is susceptible and there is a node \(v\) infectious at time step \(t\) with an edge \(uv\) with label \(t\).
Then \(u\) is infectious from time \(t+1\) until \(t+\iphase{}\), after which \(u\) becomes resistant.
Note that if a susceptible node has two or more infected neighbors at the same time,  it can be infected by any one of them, but only one.
Thus, a given set of seed infections may result in multiple possible infection chains.

The \emph{infection log} of an infection chain records which node was infected by which neighbor at what time.
Formally, the infection log is a set \(L \subseteq V^2 \times [0,\Tmax]\), where \((u,v,t) \in L\) indicates that \(u\) infected \(v\) at time step \(t\). A seed infection at \(u\) at time \(t\) is denoted by  \((u,u,t) \in L\).
The \emph{infection timetable} \(T \subseteq V \times [0,\Tmax]\) records only when a node became infected, omitting which neighbor caused the infection.
We call an infection log \(L\) \emph{consistent} with a given set of seed infections \(S\) if there is an infection chain seeded with \(S\) that produces \(L\).
Consistency for infection timetables is defined analogously.

While multiple consistent infection logs may exist for a set of seeds, there is exactly one consistent infection timetable.

\begin{appendixlemmarep}
Let \(S \subseteq V \times [0,\Tmax]\) be a set of seed infections, and \(L_1\), \(L_2\) be two infection logs consistent with \(S\). Then the induced infection timetables \(T_i \coloneqq \{(v,t) \mid (u,v,t) \in L_i\}\) (for \(i=1,2\))  are the same, that is, \(T_1 = T_2\).
\label{lem:infection-log-to-time-table}
\end{appendixlemmarep}

\begin{appendixproof}
It is easy to see that inductively, in each time step the state of all nodes must be identical under \(L_1\) and \(L_2\).
In particular, for each node, the time step in which it first becomes infected (if it ever gets infected) is the same and thus \(T_1 = T_2\).
\end{appendixproof}

\section{Modeling Temporal Graph Discovery}
\label{sec:game}\todo{ change to TGD everywhere}
We model temporal graph discovery (TGD) as a game, where the Discoverer \(\ddv\) seeks to uncover information about the graph (e.g., edges and their time labels
), while the Adversary \(\adv\) influences the infection behavior (e.g.,~by assigning edge labels) with the goal of delaying $\ddv$'s progress.
See \Cref{game:graph-discovery} for a description of the structure of the game.
\begin{figure}[tbph]
\rule{\linewidth}{0.4pt}
\textbf{Input:} \(\Tmax, \iphase, k, n \in \N\)
\begin{enumerate}[noitemsep]
\item \(\ddv\) learns the nodes and possibly additional information.\label{game:graph-discovery:step-info-sharing}
\item In each round, \(\ddv\) submits up to \(k\) seed infections $S$ and 
\(\adv\) responds with an infection log consistent with $S$. 
\label{game:graph-discovery:step-rounds}
\item To end the game, \(\ddv\) 
submits a temporal graph $\G$ to \(\adv\).
\(\adv\) responds with a temporal graph $\G'$ consistent with all infection logs. If $\G'\neq\G$, \(\adv\) wins. Otherwise, \(\ddv\) wins. \label{game:graph-discovery:step-decision}
\end{enumerate}
\vspace{-0.5\baselineskip} \rule{\linewidth}{0.4pt}
\caption{The temporal graph discovery game \label{game:graph-discovery}}
\end{figure}

As  default, we assume \(\ddv\) learns the static edges in Step \ref{game:graph-discovery:step-info-sharing}.
This is a strong, but convenient assumption, and we later show that many of our results translate to other variations.


We measure the quality of a \(\ddv\) strategy by the number of rounds it needs to win the game.
\begin{definition}
For parameters \(\Tmax{}\), \(\iphase\), \(k\) and a static graph \(G\), define the \emph{graph discovery complexity} as the minimum number of rounds required for \(\ddv\) to win any TGD game.
\end{definition}
We start off with the simplest algorithm: brute-forcing all combinations of nodes $v\in \nodeset{\G}$ and time steps $t\in[\Tmax]$ as seed infections.
This will serve as a baseline for comparing more sophisticated algorithms and possible lower bounds.
\begin{appendixtheoremrep}
There is an algorithm that wins the TGD game in \(\nodesetsize{G}\Tmax\) rounds.
\label{thm:brute-force}
\end{appendixtheoremrep}

\begin{appendixproof}
Consider the algorithm that performs the seed sets \(\{\{(v,t - 1)\} \mid v \in \nodeset{\G}, t \in [0, \Tmax-1]\}\).
Clearly, there is a successful infection along any edge at least once in these rounds.
Thus, the algorithm correctly discovers all edge labels.
\end{appendixproof}

\section{An Algorithm for Graph Discovery} \label{sec:fd-algorithm}
\label{sec:ideal-patient-zero}
Before proposing a better algorithm for TGD, we consider a simpler problem: finding an ideal patient zero.
\begin{definition}
A node \(v\) is an \emph{ideal patient zero (\ipz{}) with time \(t\)} if seed-infecting \(\{(v,t)\}\)  causes every node to become infected. We call \((v,t)\) an \emph{\ipz{} pair}.
\label{def:ideal-patient-zero}
\end{definition}
We adapt the TGD game. \(\ddv\) submits either a pair $(u,t)$, representing its guess for an \ipz{}, or $\bot$ if it believes none exists. $\adv$ responds with a graph $\G=(V,E,\edgeLabelOp)$ consistent with all infection logs. $\adv$ wins if $\ddv$'s guess is incorrect or if $\bot$ is submitted when an \ipz{} exists.

Observe that testing each node at every time it could spread an infection $(u,\edgeLabel{u,v}-1)$ for every neighbor $v$, does not always find an existing patient zero. 
Instead, brute-forcing every combination 
as in Theorem \ref{thm:brute-force}, discovers every \ipz{} pair in $\Tmax\nodesetsize{\G}$ rounds. 
\begin{corollary}
There is a strategy for \(\ddv\) that can win the \ipz{} game in \(\Tmax \nodesetsize{G}\) rounds.
\end{corollary}

\begin{figure*}[t]
\centering
\includegraphics[width=\textwidth]{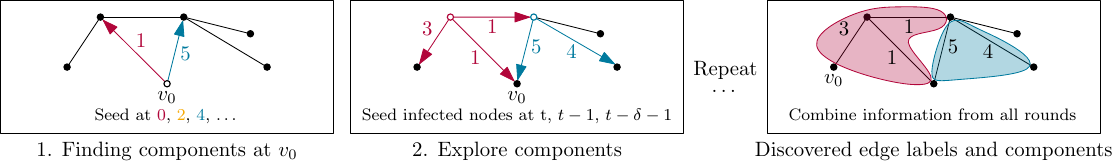}
\caption{\label{fig:follow-execution}An execution of the \(\FMainHack\) algorithm. In this example, \(\iphase = 2\). We perform the initial search for an edge at \(v_0\). Ringed nodes indicate seed infections. Next, we explore the two nodes we found in the initial search by performing seed infections at \(t\), \(t-1\), and \(t - \iphase - 1\) where \(t\) is the time step we observed the node being infected. We repeat that until we have no such seed infections to perform without repetitions. Since no \dECC{} spans all nodes, we conclude that no \ipz{} pair exists.}
\end{figure*}
The number of rounds required by the algorithm is trivially bad in its efficiency, as it is brute-forcing the entire graph instead of using the temporal graph's structure at all.
    
Note 
two edges can only participate in the same infection chain if they are connected by a series of edges with time labels differing by at most $\iphase$.
This leads to a new connectivity parameter for temporal graphs, reflecting the constraint that nodes remain infectious for only a limited time, and thus, infection chains must also respect these timing constraints.
\begin{definition}
Let \(\G\) be a temporal graph and \(\iphase \in \N^+\).
Consider the relation linking two edges if their time difference at a shared endpoint is at most \(\iphase\).
Let \(EC_\iphase\) be the partitioning of edges obtained by taking the transitive closure of this relation.
We call the resulting equivalence classes the \emph{\(\iphase\)-edge connected components} (or \dECCs{} for short) of \(\G\).
\end{definition}
Any infection chain caused by a single seed node (e.g., an \ipz{}) must be contained in a single \dECC{}.
This motivates \Cref{alg:follow}. 
See \Cref{fig:follow-execution} for an illustration of an execution.
\begin{toappendix}
Note there can be components where no infection chain infects the whole component, see \Cref{fig:delta-connected-not-infectious}.
    \begin{figure}[t]
\centering
\includegraphics[width=3.3cm]{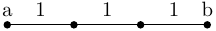}
\caption{\label{fig:delta-connected-not-infectious}
Even though (for any \(\iphase\)) \(a\) and \(b\) each have an edge in the same \dECC{}, no infection seeded at \(a\) can ever infect \(b\).}
\end{figure}
\end{toappendix}
\begin{algorithm}[tbph]
    \DontPrintSemicolon
    \SetKwFunction{FMain}{Follow}
    \SetKwFunction{ex}{Explore}
    \SetKwProg{Fn}{fun}{:}{}
    \Fn{\FMain{$G$, $\iphase$}}{
        1. Pick node $v_0 \in \nodeset{G}$ arbitrarily.

        2. For \(i \in \left [0, \left \lceil \Tmax / \iphase \right \rceil \right]\), perform a round with seed infection \((v_0, i \iphase)\).

        3. For each edge $e=v_0u$ that successfully infects: \ex{$e, \edgeLabel{e}$}.
    }
    \Fn{$\ex{u, t}$}{
    1. For each \(t' \in \{t-\iphase-1, t-1, t\}\):
        If there has been no round with seed infection $(u, t')$, seed an infection $(u, t')$. Store that this has been done.

    2. For each newly infected edge $uv$ with infection time $t$: $\ex{v,t}$.
    }
    \caption{$\FMainHack$ discovers the neighbors of $v_0$. 
    $\exHack$ 
    discovers their respective \dECCs{}.}
    \label{alg:follow}
\end{algorithm} %

The next result is the crucial property allowing us to argue that \Cref{alg:follow} does not miss relevant edges.

\begin{toappendix}
\begin{lemma}
Let \(v\) be a node seed-infected at time step \(t\) in the execution of the \(\FMain\) algorithm (\Cref{alg:follow}), and \(e\) an edge adjacent to \(v\) with label in \(t+1\) to \(t+\iphase\). Then there is a round with a successful infection via \(e\) (from \(v\) or to \(v\)).
\label{lem:follow-edge}
\end{lemma}
\begin{proof}
Let \(e = uv \in \edgeset{G}\), \(t \in [0, \Tmax]\) such that \(v\) is seed-infected at \(t\) during the execution of the \(\FMain\) algorithm.
We argue the property via downwards induction over \(\edgeLabel{e} - t\) and call this the \emph{overtaking budget}.
Intuitively, it is the time in which \(u\) could be infected by some other path than via \(e\).
Observe that if \(\edgeLabel{e} - t = 1\), the infection attempt along \(e\) must be successful as any other path to \(u\) has hop-length at least two and thus takes at least two time steps because we have strictly increasing paths.

If the infection along \(e\) is successful in the currently considered round, we are done.
So assume the infection attempt along \(e\) is unsuccessful.
Then \(u\) must be infected via a different path \(p'\) at or before \(\edgeLabel{e}\).
Let \(e' \ne e\) be the last edge on \(p'\).
By line 1 in the \(\ex\) algorithm, there must be a round where \(u\) is seed-infected at \(\edgeLabel{e'} - 1\).
Note that since \(p'\) is strictly increasing and has at least hop-length two, we have \(\edgeLabel{e'} -1 > t\). Thus, the overtaking budget for that round is strictly smaller than \(\edgeLabel{e} - t\).
Note that \(v\) and \(u\) swap roles in this recursive application, but this is immaterial as the edges are undirected.

Since the overtaking budget strictly decreases with recursive applications, it must reach 1, at which point the infection attempt along \(e\) must be successful, as argued above.
\end{proof}
\end{toappendix}

\begin{corollary}
If \(\FMain\)  discovers an edge from a \dECC{}, it discovers the whole component.
\label{cor:follow-whole-component}
\end{corollary}

This tool in hand, we prove the correctness of \Cref{alg:follow}, which requires at most \(6 \edgesetsize{G} + \left \lceil T_\mathrm{max}/\iphase \right \rceil\) rounds.
\begin{appendixtheoremrep}
The \(\FMain\) algorithm (\Cref{alg:follow}) correctly solves the \ipz{} problem.
\label{thm:follow-correct}
\end{appendixtheoremrep}

\begin{proofsketch}
    Observe that Step 2 always finds all edges adjacent to \(v_0\) and by \Cref{cor:follow-whole-component}, the \(\exHack\) algorithm discovers the whole \dECCs{} of these edges.
    Finally, we show that if an \ipz{} exists, their infection chain must be subset of one of these components.
\end{proofsketch}

\begin{appendixproof}
Let \(v_0\) be the start vertex picked in the algorithm.
Observe that the loop in Step 2 of \(\FMain{}\) discovers at least one edge from each \dECC{} adjacent to \(v_0\).
Applying \Cref{cor:follow-whole-component}, we can see that then the algorithm discovers all edges in these components.
To finish the proof we are left to show that if we know the labels of edges that are in the same \(\iphase\)-connected edge component as any edge at \(v_0\), we can find an \ipz{} if it exists.

Assume that there is an \ipz{}.
Now, take any infection chain caused by seed-infecting this \ipz{} pair.
We construct the \emph{directed tree of successful infections} by only including edges along which there was a successful infection and by directing all edges in the direction along which the infection traveled.
By definition of the \ipz{} pair, this tree spans the entire graph.
Thus, one of its edges is adjacent to \(v_0\).
By definition, the directed tree of infections must be in a single \(\iphase\)-connected edge component.
Therefore, all edge labels of the tree as well as all other edge labels that could affect the outcome of the \ipz{} infection (edges in different components do not interact if there is only one seed infection) are known at the end of the algorithm.

Since we learn all relevant edges through the algorithm, we can decide the existence of an \ipz{}.
\end{appendixproof}
\begin{toappendix}
\begin{theorem}
The \(\FMain\) algorithm terminates after at most \(6 \edgesetsize{G} + \left \lceil T_\mathrm{max}/\iphase \right \rceil\) rounds.
\end{theorem}

\begin{proof}
The search of edges adjacent to \(v_0\) takes \(\left \lceil T_\mathrm{max}/\iphase \right \rceil\) rounds.
The \(\ex\) sub-algorithm uses at most 6 rounds per edge \(e\) (at most \(\edgeLabel{e}\), \(\edgeLabel{e}-1\), and \(\edgeLabel{e}-\iphase-1\) per endpoint).
This yields the desired bound.
\end{proof}
\end{toappendix}

\label{fd-algorithm}
We now extend this idea to obtain a better TGD algorithm, see \Cref{alg:follow-discover}.
Note that its running time does not only depend on the static, but also the temporal structure of the graph.
Recall that \(\FMain\) explores precisely the \dECCs{} adjacent to the start node \(v_0\) and 
note that a graph is discovered if and only if all its \dECCs{} are discovered. 

    \begin{algorithm}[tbhp]
        \DontPrintSemicolon
        \SetKwFunction{FDMain}{DiscoveryFollow}
        \SetKwFunction{ex}{Explore}
        \SetKwProg{Fn}{fun}{:}{}
        \Fn{\FDMain{$G$, $\iphase$}}{
            \While{there is a node $v_0$ with an adjacent edge for which the label is still unknown}{
                In steps of size $\iphase$ perform rounds such that $v_0$ is seed-infected at $0, \iphase, 2\iphase, \dots, \Tmax$.

For each edge $e=v_0u$ which successfully infects from $v_0$ to a neighbor: $\ex{$u$, $\edgeLabel{e}$}$.
            }
        }
    \caption{The TGD extension of \(\FMainHack \).}
    \label{alg:follow-discover}
    \end{algorithm}

\begin{appendixtheoremrep}
\Cref{alg:follow-discover} wins any instance of the TGD game on \(\G\) in at most \(6\edgesetsize{\G} + \decc{\G} \left\lceil \Tmax/\iphase \right \rceil\) rounds.
\label{thm:follow-discover}
\end{appendixtheoremrep}

\begin{appendixproof}
The correctness follows since, by \Cref{cor:follow-whole-component}, every call of the \(\ex\) algorithm discovers all edges which share a \dECC{} with one edge adjacent to the start vertex (i.e.,~the one picked as \(v_0\) in step 1).
Since we explore all \dECCs{}, we discover all edges.

For the running time, observe that the loop in  \cref{alg:follow-discover} runs at most \(\decc{\G}\) times.
Since by the same argument as before, there are at most \(6\) infections per edge (in essence we avoid duplication), the factor only applies to the second summand, yielding the stated bound.
\end{appendixproof}

\section{Lower Bounds for Graph Discovery} \label{sec:witnesses}
We build a toolkit for proving lower bounds in the TGD game.
Initially, every edge could have any label. As the game progresses, information is revealed to $\ddv$. 
A successful infection attempt determines the edge's label, while an unsuccessful attempt where one endpoint remains susceptible reduces the possible labels by at least one.
With this, we can use a potential argument to establish lower bounds on TGD complexity.

\begin{appendixtheoremrep}
Let \(\G\) be a temporal graph and \(T_\mathrm{max}\), \(\iphase\), \(k\) parameters for the TGD game. For a sequence of seed infection sets \(S_1, \dots, S_a\), define \(\Phi(i)\) as the sum of the sizes of the sets of consistent labels over all edges after rounds 1 to \(i\). Then, \(S_1, \dots, S_a\) is a \emph{witnessing schedule} if \(\Phi(a) = \edgesetsize{\G}\).
\label{thm:potential-argument}
\end{appendixtheoremrep}

\begin{appendixproof}
Clearly, if there is more than one consistent label for some edge, \(\adv\) will always win step \ref{game:graph-discovery:step-decision} in the TGD game (\cref{game:graph-discovery}).
\end{appendixproof}

Recall that the brute-force algorithm for TGD takes \(\BigO(n \Tmax)\) rounds (\Cref{thm:brute-force}).
In the worst case, this can be improved by at most a factor of \(\iphase k\).
\begin{theorem}
For all \(\iphase, k \in \N^+\) and \(\Tmax \ge 4\), there is an infinite family of temporal graphs \(\{\G_n\}_{n \in \N}\) such each \(\G_n\) has \(\Theta(n)\) nodes and the minimum number of rounds required to win the TGD game on it grows in \(\Omega(n (T_\mathrm{max} - 3) / (\iphase k))\). 
\label{thm:temporally-connected-graph-discovery-lower-bound}
\end{theorem}


\begin{proof}
Let \(n, \iphase \in \N^+\) and \(\Tmax \ge 4\), with $n$ even for simplicity. We construct a temporal graph $\G_n$ on vertices $v_1,\dots,v_n$. The edges are as follows (with some edge labels already given which \(\adv\) will always assign).
First, \(v_1, \dots, v_{n-2}\) form a path where every even-indexed edge has label \(\Tmax\),
Second, \(v_{n-1}\) has an edge to each \(v_1, \dots v_{n-2}\) with fixed label \(T_\mathrm{max}-2\),
\(v_{n-1}\) and \(v_{n}\) share an edge with fixed label \(T_\mathrm{max}-1\), and
\(v_{n}\) has an edge to each \(v_1, \dots v_{n-2}\) with fixed label \(T_\mathrm{max}\).

Assume \(\ddv\) acts arbitrarily. 
%
For edges with fixed labels, $\adv$ responds that label.
For all other edges, $\adv$ replies ``infection failed'' as long as at least one consistent label would remain; otherwise, it replies with any consistent label.

We apply the potential argument from \Cref{thm:potential-argument} to all edges that do not have a fixed label. We call those \((n-3)/2\) edges \emph{relevant}.
Initially, each relevant edge has $\Tmax$ possible labels, so \(\Phi(0) \ge (n-3)/2 \cdot T_\mathrm{max}\).
In each round, for every node that gets infected, at most $\iphase$ labels are removed per adjacent relevant edge, either due to failed infection attempts, or a successful infection with $\iphase$ or fewer labels remaining. Thus, the potential decreases by at most $\iphase k$ per round.
Dividing the initial potential by this maximum decrease shows that any $\ddv$ needs at least
$
\left \lfloor (n\cdot (T_\mathrm{max}-3))/(2 \iphase k) \right \rfloor
$
\todo{why did the -3 move from n to \(T_{max}\)?}
rounds. 
\todo{You should probably mention that infections never actually spread further than 1 edge to say that you really only get information about one edge. And you should take the last two labels out of the potential because for those you can actually get information from a spread infection.}
\end{proof}

\subsection{Witness Complexity}
\label{sec:org66bff96}
The witness complexity of a temporal graph is the minimum number of rounds required for a \(\ddv\), knowing all labels, to convince an observer of the labeling's correctness. 
\begin{definition}
A length \(a\) \emph{witnessing schedule} for a temporal graph \(\G\) is a sequence of seed infection sets \(S_1, \dots, S_a\) such that after performing \(a\) rounds with the respective seed infection sets, all labels in the graph are uniquely determined by the logs of these rounds.
The \emph{witness complexity} of a temporal graph \(\G\) is the length of its shortest witnessing schedule.
\end{definition}
Observe that the graphs constructed in the proof of \Cref{thm:temporally-connected-graph-discovery-lower-bound} have witness complexity \(O(n)\), which is significantly lower than their graph discovery complexity.
However, witness complexity is a powerful tool for establishing lower bounds on graph discovery complexity, particularly when seeking bounds independent of \(\Tmax\). The following lemma states the formal relationship between the two complexities.
\begin{appendixlemmarep}
Let \(\G\) be a temporal graph and \(T_\mathrm{max}\), \(\iphase\), \(k\) parameters as defined above. Then the witness complexity of \(\G\) in this instance is at most as large as the graph discovery complexity for the same parameters.\todo{I rephrased it slightly. Does it make sense now?}
\end{appendixlemmarep}
\begin{appendixproof}
This follows directly from the fact that recording the seed infection sets of any online algorithm yields a witnessing schedule upon termination.
\end{appendixproof}

Note however, that the witness complexity technique can only ever be at most the number of edges in a graph.
\begin{appendixtheoremrep}
For any instance \((\G, \Tmax{}, \iphase, k)\), the witness complexity is at most \(\edgesetsize{\G}\).
\end{appendixtheoremrep}

\begin{appendixproof}
Let \(e_1, \dots, e_{\edgesetsize{\defaultGraphName}}\) be an arbitrary numbering of the edges, and set \(S_i \coloneqq \{(e_i, \edgeLabel{e_i}-1\}\).
This is a witnessing schedule of length \(\edgesetsize{\defaultGraphName}\).
\end{appendixproof}

We now show that this worst case is actually tight, and there are graphs that asymptotically require about one round per edge to witness correctly.
Observe that this bound is irrespective of \(k\),  thus does not have one of the major shortfalls of our previous lower bounds for the TGD game.
\begin{appendixtheoremrep}
There is an infinite family of temporal graphs \(\{\G_n\}_{n \in \N}\) whose witness complexity grows in \(\Omega_k(\edgesetsize{\G_n})\).
\label{thm:omega-m-witness-complexity}
\end{appendixtheoremrep}
We will now define \(\{\G_n\}_{n \in \N}\). Then we formalize when an infection attempt is \emph{relevant}, in the sense that it makes meaningful progress towards winning the witness complexity game.
Lastly, we show that there can be at most one such relevant infection attempt per round in \(\{\G_n\}_{n \in \N}\).
    
Intuitively, each graph in \(\{\G_n\}_{n \in \N}\) contains four node sets: $L$, $R$, $B$, and $C$. The nodes in $L$ and $R$ form a complete bipartite graph, where the edges between a node in $L$ and all nodes in $R$ are assigned to distinct \textit{phases}. 
The nodes in $R$ are connected such that once a node in $R$ is infected, it spreads the infection through $R$ without additional input from $L$.
The sets $B$ and $C$ serve as gadgets to ensure that at the end of each phase, all nodes in $L$ and $R$ become infected via edges outside $L$ and $R$. This prevents further information being gathered about the labels between $L$ and $R$, ensuring the lower bound on the complexity of the discovery process.

For $x\in \N$, we construct a temporal graph $\G_n$ of size $n\coloneqq 5x$. The vertices are given by $V(\G_n)\coloneqq L \sqcup R \sqcup B \sqcup C$ with $L \coloneqq \{\ell_1, \dots, \ell_x\}, R \coloneqq \{r_1, \dots, r_{2x}\}, B \coloneqq \{b_1, \dots, b_x\}, C \coloneqq \{c_1, \dots, c_x\}.$ Denote by $R_2$ the nodes in $R$ with an even index.
Then the edges are given by $E(\G_n)\coloneqq L \times R_2 \sqcup R^2 \sqcup L\times B \sqcup B\times R\sqcup B^2 \sqcup C \times R \sqcup \{b_ic_i\mid i \in [x]\}$.
The labels on the edges are defined as follows:

    Let \(P = \{p_1, \dots, p_x\}\) be a set of \(x\) edge-disjoint Hamiltonian paths on \(R\). Such paths exists by \cite{axiotis_approx}.
    Note that by their construction, for every path, there are at most two nodes with an odd index in a row.
    Assume without loss of generality that each \(p_j\) begins at \(r_{2j}\), and write \(r_{p(i,j)}\) for the \(j\)-th node from \(R_2\) on the path \(p_i\). Denote by \(p^{-1}(i,j)\) its index on \(p_i\).
    Let \(k\) be arbitrary, \(T_\mathrm{max} \coloneqq x(4x+1)\), and \(\iphase \coloneqq 4x+1\).
    \allowdisplaybreaks
    Now,
    \begin{align*}
    &\text{for }\ell_i \in L, j \in [x], \text{ set } 
    &&\hspace{-0.7em}\edgeLabel{\ell_i, r_{p(i,j)}} \coloneqq i\iphase + 4(p^{-1}(i,j))+1,  \\
    &\text{for } p_q \in P, j\in[x], \text{ set } 
    &&\hspace{-0.7em}\edgeLabel{(p_q)_j, (p_q)_{j+1}} \coloneqq q\iphase + 2j+1,  \\
    &\text{for } b_i \in B, j \in [x], \text{ set } 
    &&\hspace{-0.7em}\edgeLabel{b_i, r_{(p_i)_j)}} \coloneqq i\iphase + 2j + 2,  \\
    &\text{for } b_i, b_j \in B, i<j, \text{ set } 
    &&\hspace{-0.7em}\edgeLabel{b_i, b_j} \coloneqq (i+1)\cdot\iphase -2, \\
    &\text{for } i \in [x] \text{ set } 
    &&\hspace{-0.7em}\edgeLabel{b_i, c_i} \coloneqq (i+1)\cdot\iphase - 1, \\
    &\text{for } i > j \in [x] \text{ set } 
    &&\hspace{-0.7em}\edgeLabel{b_i, c_j} \coloneqq (i+1)\cdot\iphase - 2, \\
    &\text{for } c_i \in C, r_j \in R \text{ set } 
    &&\hspace{-0.7em}\edgeLabel{c_i, r_j} \coloneqq (i+1)\cdot\iphase, \\
    &\text{for } \ell_i \in L, b_j \in B, \text{ set } 
    &&\hspace{-0.7em}\edgeLabel{\ell_i, b_j} \coloneqq (j+1)\cdot\iphase.
    \end{align*}

    Observe that the temporal edges of $\G_n$ can be partitioned into \(x\) sets, each having labels in a fixed interval of size \(\iphase\) and for \(i \in [x]\), let \(E_i = \{e \in E(\G_n) \mid \edgeLabel{e} \in [\iphase \cdot i, \iphase \cdot (i+1)-1]\}\).
    We refer to the edges $E_i$ and corresponding intervals as \emph{phases}.
Now, an infection attempt between some \(l_i \in L\) and \(r_{2j} \in R_2\) is called \emph{relevant} if (i) it happens at \(\edgeLabel{l_i r_{2j}}\) and is successful, or (ii) it happens at \(\edgeLabel{l_i r_{2j}} - 1\), is unsuccessful, and exactly one endpoint was infected before the attempt.

The result then follows from the following three properties: (1) for each edge in \(L \times R_2\), there has to be at least one relevant infection attempt to win the witness complexity game, (2) there is at most one relevant infection per phase, and (3) there is at most one phase with a relevant infection per round.

\begin{toappendix}
    
\begin{lemma}
If, at the end of the witness complexity game, the Prover wins, there must have been one relevant infection attempt for every edge in \(L \times R_2\).
\label{lem:number-relevant-attempts}
\end{lemma}

\begin{proof}
The result follows since if there is an edge in \(L \times R_2\) for which there was no relevant infection attempt, then both \(\edgeLabel{e}\) and \(\edgeLabel{e} - 1\) are possible labels consistent with the infection log.
Therefore, the Prover must lose the game, since \(\adv\) can always pick the label the Prover does not pick out of these two.
\end{proof}

\begin{lemma}
In each round, there is at most one phase with relevant infection attempts.
\label{lem:active-phases-per-round}
\end{lemma}

The proof of this statement is fairly technical. We argue that all infection chains must essentially follow the pattern illustrated in \Cref{fig:construction-omega-m}.
Our construction contains gadgets to ensure this also happens in all edge cases, which we carefully consider in the proof to make this idea rigorous.

\begin{figure}[tb]
\centering
\includegraphics[width=4cm]{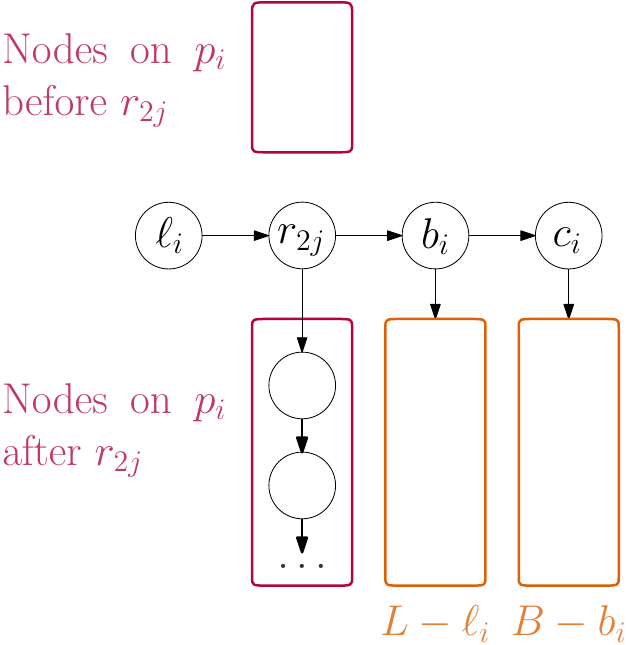}
\caption{\label{fig:construction-omega-m}As an illustrative example, consider the directed tree of infections (an edge \(u \to v\) means that \(v\) is infected via \(uv\)) that occurs when \(\ell_i\) is infected such that the relevant infection attempt is along \(\ell_ir_{2j}\). The other cases are similar. Colored arrows mean that all nodes in a set are infected via an edge from the origin node.\mm{there are no colored edges?}}
\end{figure}

To prove this lemma, we need to formalize the notion of an active phase, which, in essence, is the phase where the relevant infection attempts must take place.

\begin{definition}
Let \(v\) be the first node infected in a given round, let \(t\) be the infection time, and let \(i\) be the phase of that infection time.
Then one of two things must be the case:
\begin{enumerate}
\item At time step \((i+1)\iphase\) all nodes in \(L\) and \(R_2\) are infected or resistant, and thus no relevant infection attempt may take place after that time. In this case, we call \(i\) the active phase.
\item At time step \((i + 1) \iphase\) not all nodes in \(L\) and \(R_2\) are infected or resistant. In that case, call \(i+1\) the active phase.
\end{enumerate}
\end{definition}

\begin{proof}
In what follows, we prove that there is at most one active phase and that any two relevant infection attempts must occur within that active phase.
For technical reasons, the first infection in a round is either in the active phase or right before the start of the active phase.

Towards that end, we prove that \(b_i\) is infectious at time steps \((i+1)\iphase\) and \((i+1)\iphase-1\) or \(b_{i+1}\) is infectious at \((i+2)\iphase\) and \((i+2)\iphase -1\), but in the second case, there are no infection attempts before \(i\iphase + 4x + 1\) (and thus no relevant infection attempts).

Note that for all infection chain descriptions below, we can ignore the fact that, when we claim a node infects another, there may have been additional seed infections which lead to infecting one of the endpoints earlier. Since we know that there are no infections before \(i\iphase\) and are only interested in infections until \((i+1)\iphase\), earlier infections still leave the nodes infectious at the claimed time steps anyways, except in those cases where we explicitly argue about scenario two.

\uline{Case 1:} \(v = b_i\). If \(t \le i\iphase + 4x + 1\), \(b_i\) infectious at both \((i+1)\iphase\) and \((i+1)\iphase -1\). If not, case two clearly holds. There will be no relevant infections (because it’s too late) and \(b_i\) infects \(b_{i+1}\) at time step \((i+2)\iphase - 2\) which is then infectious at \((i+2)\iphase\) and \((i+2)\iphase -1\).

\uline{Case 2:} \(v \in B-b_i\). If \(t \le i\iphase + 4x + 1\), then \(v\) infects \(b_i\) at time step \((i+1)\iphase - 2\) and the statement holds. Otherwise, argue scenario two analogously to case 1.

\uline{Case 3:} \(v \in L\). If \(v=\ell_i\) and the infection occurs before \(i\iphase + 4x + 1\), then \(\ell_i\) infects some node in \(R_{2}\) within the next four time steps, which then infects \(b_i\) right after. In all other cases, \(v\) infects \(b_i\) at time step \((i+1)\iphase\) which then infects \(b_{i+1}\) at time step \((i+2) \iphase -2\) which then infects \(c_{i+1}\) at time step \((i+2)\iphase -1\), so that at time step \((i+2) \iphase\) all nodes in \(L\) and \(R_2\) become either infected or are already resistant or infected.

\uline{Case 4:} \(v \in R\). Assume \(v = r_{(p_i)_j}\). Then, if the seed infection occurs before \(i\iphase + 2j + 2\), \(b_i\) gets infected at that time via the edge \(r_jb_i\). If the infection occurs after \(i\iphase + 4j + 2\), we have to analyze a bit more carefully. Now, \(c_i\) becomes infected at time step \((i+1)\iphase\).  Then, \(c_i\) infects \(b_{i+1}\) at \((i+2)\iphase-2\), fulfilling our condition.

\uline{Case 5:} \(v \in C\). Let \(v = c_j\) with \(j \in [x]\). If the infection occurs before \((i+1)\iphase -1\), then \(b_i\) becomes infected at \((i+1)\iphase - 1\). If the infection occurs after \((i+1)\iphase -1\), then \(b_{i+1}\) becomes infected at \((i+2)\iphase -1\). If the infection is this late, clearly, no relevant infection can occur before \((i+1)\iphase\), that is the end of phase \(i\).

We have seen that if the seed infection(s) occur before \(i\iphase + 4x + 1\), \(b_i\) is infectious at time steps \((i+1)\iphase - 1\) and \((i+1)\iphase\). In this case we call \(i\) the first active phase. In the other case, \(b_{i+1}\) is infectious at \((i+2)\iphase - 1\) and \((i+2)\iphase\) and we call \(i+1\) the first active phase.

Let \(i'\) be the first active phase, then at time step \((i'+1)\iphase -1\) the node \(b_{i'}\) infects \(c_{i'}\) (or \(c_{i'}\) is already infectious). Then at time step \((i'+1)\iphase\), all nodes in \(L\) and \(R\) become infected (if they are not infected or resistant before). Thus, after that time step, no relevant infection attempt may occur anymore.
\end{proof}

\begin{lemma}
There is at most one relevant infection attempt per phase.
\label{lem:relevant-attempts-per-phase}
\end{lemma}

\begin{proof}
Let \(i\) be the first active phase.
Then any relevant infection attempt must be made in connection with \(\ell_{i}\) since the time steps related to relevant infection attempts for other nodes in \(L\) are all either in earlier (inactive) phases or in later phase when, as proven above, no relevant infection attempt may take place.

Now assume there are two relevant infection attempts along two different edges \(e\) and \(e'\).
Write \(e = \ell_i r_{2j}\) and \(e' = \ell_i r_{2j'}\) and assume without loss of generality that \(\edgeLabel{e} < \edgeLabel{e'}\).

Observe that since \(\edgeLabel{e} < \edgeLabel{e'}\), it follows that \(r_{2j'}\) must be further back on \(p_i\).
If the infection attempt along \(e\) is relevant because it’s successful, then \(r_{2j'}\) is infected at the latest via \(p_i\) one time steps before a relevant infection attempt can occur via \(e'\).

Now assume that there is a relevant infection attempt along \(e\) but there is no successful infection via \(e\).
This could only happen if one of two things were the case: (1) \(\ell_i\) is infectious at \(\edgeLabel{e} -1\), but not at \(\edgeLabel{e}\) or (2) \(r_{2j}\) is already infected or resistant at \(\edgeLabel{e}\).

(1) cannot happen because then \(\ell_i\) must have been infectious before \((i-1)\iphase + 4x + 1\), which contradicts our assumption.
Regarding, (2) if \(r_{2j}\) is infectious long enough to infect their next node on \(p_i\), the argument proceeds as above.
If \(r_{2j}\) is not infectious anymore by that time, it must have been infectious before \((i-1)\iphase + 4x + 1\), again contradicting our assumption.
\end{proof}

\begin{proof}[Proof of \cref{thm:omega-m-witness-complexity}]
Putting together \cref{lem:number-relevant-attempts,lem:active-phases-per-round,lem:relevant-attempts-per-phase}, we get the desired bound.
\end{proof}
\end{toappendix}

\section{Extending the Model} \label{sec:extending}
So far, we examined a simple version of the TGD game with restricted assumptions about infection behavior and the knowledge of \(\ddv\). These might seem restrictive and less close to the real-world processes. In this section, we lift these restrictions and show that either the theoretical behavior remains unchanged, or the problem becomes trivial, offering no significant improvement over brute-force. 
See \Cref{tbl:variations-overview} for an overview of the resulting lower bounds and algorithms. 


\ifpaper \graphDiscoveryResultsTable \fi

\begin{toappendix}
\subsection{Infection Times Only–Feedback \label{sec:infection-time-feedback}}
\label{sec:orgc0595d3}
First, let us look at a variation of the TGD game where the player only gets feedback on when a node is infected, but not by which other node.
Call this \emph{infection times only–feedback}.
This makes \(\ddv\)'s job harder.
In some cases (e.g.,~when a node only has one neighbor), \(\ddv\) can still deduce who infected a node, but generally that is not the case anymore.
To see this, consider a case where a node becomes infected at some time when it has two infectious neighbors. Then, one of the edges must have the label of the infection time, but \(\ddv\) cannot directly deduce which of the edges, as it could in our basic model.
Since, up until now, we have looked at an easier (from \(\ddv\)'s perspective) version of the TGD game, lower bounds directly translate.
We will see this pattern for the other variations as well, though we will not state it with this level of formality from now on.

\begin{theorem}
Let \((G_n, (\Tmax)_n, \iphase_n, k_n)_{n \in \N}\) be a family of instances of the TGD game under the basic model and \(f: \N^5 \to \N\) such that any \(\ddv\) must use at least
\(\Omega(f(\nodesetsize{G_n}, \edgesetsize{G_n}, (\Tmax)_n, \iphase_n, k_n))\)
rounds, then the same lower bound holds if the game is played under the infection times only–variation.
\end{theorem}

\begin{proof}
The result follows since \(\ddv\) gets strictly less information, and the rest of the process is exactly the same.
Thus, every algorithm winning the game under the infection times only–variation must also win the game under the basic model.
\end{proof}

Crucially, applying this to \Cref{thm:temporally-connected-graph-discovery-lower-bound} and \Cref{thm:omega-m-witness-complexity} yields the following two results.

\begin{corollary}
For all \(\iphase, k \in \N^+\) and \(\Tmax \ge 4\), there is an infinite family of temporal graphs \(\{\G_n\}_{n \in \N}\) such that the graph \(\G_n\) has \(\Theta(n)\) nodes and the minimum number of rounds required to satisfy the TGD game under the infection times only–variation grows in \(\Omega(n (T_\mathrm{max} - 3) / (\iphase k))\). Also, all graphs in the family are temporally connected.
\label{thm:temporally-connected-graph-discovery-lower-bound-infection-times-only}
\end{corollary}

\begin{corollary}
There is an infinite family of instances \(\{\G_n\}_{n \in \N}\) such that the witness complexity  under the infection times only–variation grows in \(\Omega_k(\edgesetsize{\G_n})\).
\end{corollary}

The trivial algorithm still works with infection times only-feedback. What is not quite as obvious is that the \(\FDMain\) TGD algorithm also translates, giving us the following result.

\begin{appendixtheoremrep}
\Cref{alg:follow-discover} wins any instance of the TGD game under the infection times only–variation on a  graph \(\G\) in at most \(6\edgesetsize{\G} + \decc{\G} \left\lceil \Tmax/\iphase \right \rceil\) rounds.
\end{appendixtheoremrep}

\begin{appendixproof}
First, observe that the algorithm does not require the source of an infection to be known. Calling the \(\ex\) subroutine simply needs the fact that a node was infected and at what time.

Secondly, observe that at the end, for each edge \(e = \{u,v\}\) there has been a seed infection at \(u\) and \(v\) at \(\edgeLabel{e} - 1\), thus the infection must have been successful and taken place in the first possible time step (after the seed infection).
As any infection chain not along the single edge between \(u\) and \(v\) takes at least two time steps, we can soundly conclude the label of the edge from the infection logs.

Together, these two properties mean that the \(\FDMain\) algorithm translates to the infection times only variation.
\end{appendixproof}

\subsection{Unknown Static Graph \label{sec:unknown-static-graph}}
\label{sec:org677ea02}
It is a fair criticism that it is not always realistic to assume \(\ddv\) knows the static graph at the start of the game.
This motivates us to investigate a version of the game where only the node set, but not the edges, are known to \(\ddv\) at the start of the game.
Here too, the lower bounds from the basic model translate as \(\ddv\) gets strictly less information, but we can also achieve new, stricter bounds, which show that not only does our \(\FDMain\) not translate, no comparably fast algorithm can even exist.

\begin{theorem}
Consider the variation of the TGD game where \(\ddv\) does not learn the static graph. Let \(n, T_\mathrm{max} \in \N^+\) and \(k \in [n], m \in [{n \choose 2}-n], \iphase \in [T_\mathrm{max}]\). Then there is an \(\adv\) such that any algorithm winning this game variation on graphs with \(n\) nodes must take at least \(\lfloor nT_\mathrm{max}/(2\iphase k) \rfloor\) rounds. This \(\adv\) picks a graph with at most two \dECCs{}.
\label{thm:unknown-static-graph-lb}
\end{theorem}

Note that this result holds regardless of the number of edges in the graph and for just two \dECCs{}.
Therefore, we cannot hope for an algorithm only dependent on the number of edges and \dECCs{} (as we have in the basic model).
In the basic model, nodes without edges are the best case, as we can simply ignore them.
Here, the opposite is true, since we perform tests to ensure the non-existence of the edge at every time step.

\begin{proof}
As the underlying graph, we fix all \(m\) edges but one, which we pick dynamically.
Choose the \(m -1\) fixed edges arbitrarily such that every node has at most \(n-2\) adjacent edges and that their subgraph is connected.
This is possible by a simple greed strategy starting at a single node \(v\), connecting it to arbitrary nodes that have less than \(n-2\) neighbors until \(v\) has  \(n-2\) neighbors, and then moving on to one of these neighbors to do the same.
Repeat this process until \(m-1\) edges have been added.
Give those edges the time label 1.
For these, \(\adv\) responds to infection attempts by simply simulating the correct behavior (i.e.,~an attempt is successful precisely at time step 1).

For all other possible edges, respond to all infection attempts with ‘infection failed’ as long as after that response there is a still are \(u,v \in \nodeset{G}\) and \(t \in [\Tmax]\) such that
\begin{itemize}
\item \(uv\) is not one of the edges fixed before, and
\item there has been no unsuccessful infection attempt from \(u\) to \(v\) or \(v\) to \(u\) at \(t\).
\end{itemize}
Otherwise, respond with `infection successful`.
Clearly, no \(\ddv\) can terminate and win until one round before that has happened, since there are still multiple consistent options for the unfixed edge remaining.

A similar argument to \Cref{thm:temporally-connected-graph-discovery-lower-bound} shows that, since no infection can ever spread, the online algorithm essentially has to search through all nodes \todo{technically you just need through a vertex cover of the non-edges right? And I think by your construction that can actually be quite small? You kinda want a perfect matching with non-edges right?} at all time steps, working on at most \(k\) nodes per round and covering at most \(\iphase\) time steps per node.
This yields the \(\lfloor n \Tmax / (2\iphase k) \rfloor\) bound.
\end{proof}

Therefore, while the goal of this variation seems natural, it makes the problem difficult to efficiently solve.
In particular, exploiting the structure of the static graph is hard, as both edges and non-edges need to be proven.
Naturally, this means that any improvement over the brute-force method is limited.
On the positive side, the brute force method still works.

\begin{appendixtheoremrep}
There is an algorithm that wins the TGD game in the unknown graphs variation in \(\nodesetsize{G}\Tmax\) rounds. \todo{Again I would state the algorithm in the theorem statement}
\end{appendixtheoremrep}

\begin{appendixproof}
As for \Cref{thm:brute-force}, the algorithm that performs the seed sets \(\{\{(v,t - 1)\} \mid v \in \nodeset{G}, t \in [0, \Tmax-1]\}\) trivially wins the game.
\end{appendixproof}

While, in that model \todo{which model?}, we were able to find a better algorithm for sparse graphs (in particular the \(\FDMain\) algorithm), we have no such hope here by \Cref{thm:unknown-static-graph-lb}.
In fact, these results show that this naive algorithm is close to optimal.

\subsection{Multilabels \& Multiedges \label{sec:multigraphs}}
\label{sec:org8b2260a}
The last extension we investigate is what happens if an edge between the same two nodes may exist at multiple time steps.
There are two ways to model this: we could allow each edge to have a set of labels instead of just one (we call these \emph{multilabels}) or we could allow multiple distinct edges (with different labels) to exist between the same two nodes (we call these multiedges).
Note that, while for most problems in the literature these notions are equivalent, that is not the case here.
With multiedges, \(\ddv\) learns the multiplicity of every edge, but with multilabels it does not.
If we only tell \(\ddv\) where an edge is, but not how many time labels it has, then we run into the same issue as with unknown static graphs in \Cref{sec:unknown-static-graph}.
In essence, any \(\ddv\) would be forced to check all combinations of nodes and seed times, only allowing for the trivial algorithm.
We first formalize this notion by giving the appropriate lower bound for multilabels before giving more positive results for multiedges.

\begin{appendixtheoremrep}
Consider the variation of the TGD game where a single edge might have an arbitrary subset of \([\Tmax]\) as labels. Let \(n, T_\mathrm{max} \in \N\) and \(k \in [n], m \in [{n \choose 2}-n], \iphase \in [T_\mathrm{max}]\). Then there is an \(\adv\) such that any algorithm winning this game variation on graphs with \(n\) nodes must take at least \(\lfloor \min\{n/2, m\}\Tmax/(\iphase k) \rfloor\) rounds. This \(\adv\) picks a graph with at most two \dECCs{}.
\label{thm:multilabels-lb}
\end{appendixtheoremrep}

Again, this lower bound tells us we may not hope for an algorithm taking advantage of a small number of \dECCs{}.
Similarly, any algorithm can only take advantage of the knowledge of the static edges when there are few of them (precisely if the number of edges is sublinear in the number of nodes).

\begin{proofsketch}
 The proof is analogous to the one of \Cref{thm:unknown-static-graph-lb} with three minor modifications: (1) we must fix all edges and let \(\ddv\) search for possibly existing labels instead of edges, (2) we must assure that any node has few adjacent edges, and (3) the analysis of the potential now analyzes the search for the multilabels and respects this smallest vertex cover.
We need to bound the number of edges adjacent to any one node to ensure that we cannot check all the edges using a small number of nodes to perform seed infections at.
Observe that this condition implies that the smallest vertex cover is large (i.e.,~we need many nodes to cover all edges).
\end{proofsketch}

\begin{appendixproof}
The proof is analogous to the one of \Cref{thm:unknown-static-graph-lb} with three minor modifications.
For ease of reading, we repeat the three modifications as outlined in the proof sketch:
(1) we must fix all edges and let \(\ddv\) search for possibly existing labels instead of edges, (2) we must assure that any node has few adjacent edges, and (3) the analysis of the potential now analyzes the search for the multilabels and respects this smallest vertex cover.
We need to bound the number of edges adjacent to any one node to ensure that we cannot check all the edges using a small number of nodes to perform seed infections at.
Observe that this condition implies that the smallest vertex cover is large (i.e.,~we need many nodes to cover all edges).

To address (1) and (2), construct the edges of the graph greedily by starting at an arbitrary node and connecting it to a different arbitrary node.
Now, until you have added \(m\) edges, repeatedly consider the node that has the smallest positive number of adjacent edges and add an edge to the node with the smallest number of adjacent edges (0 if possible).
Clearly, in the resulting graph, all nodes that have at least one edge are in the same connected component, and thus, after labeling them \(1\),  all these edges are in the same \dECC{}.
There will be at most one additional label assigned by \(\adv\), again yielding at most two \dECCs{} in total.
Also, any node has degree at most \(\lceil m /n \rceil\), thus any vertex cover must have at least \(\lceil m / (\lceil m / n \rceil) \rceil\) nodes.
Note that this implies that if \(m \ge n\), the size of the minimum vertex cover is at least \(\lfloor n/2 \rfloor\).

Regarding (2), observe that, initially, there are \(m \Tmax\) possible edge labels (call that the initial potential).
In this variation, a successful infection only tells \(\ddv\) that the edge has that label but not that it does not have other labels, as is the case in the basic game.
Also, by the construction of our graph, the smallest vertex cover is large, so any seed infection can only test a small number of edges.
Concretely, any node has at most \(\lceil m / n \rceil\) adjacent edges.
Therefore, any round can decrease the potential by at most \(\lceil m / n \rceil \iphase k\).
Dividing \(m \Tmax\) by \(\lceil m / n \rceil \iphase k\) yields the desired bound.
\end{appendixproof}

The situation looks much better if we consider temporal multiedges (i.e.,~temporal multigraphs where there can be multiple distinct edges between the same two nodes) instead.
Here, the core advantage we have is that we know the multiplicity of the edges.
In the proof of \Cref{thm:multilabels-lb} we saw that the crucial property that forces \(\ddv\) to spend so many rounds is that discovering a label on an edge does not preclude it from having to check all other possible time steps for more labels.
Note that we disallow multiple edges with the same endpoints and the same label.

First, notice that any temporal graph is also a temporal multigraph.
Therefore, we have the following lower bound as a corollary to \Cref{thm:temporally-connected-graph-discovery-lower-bound}.

\begin{corollary}
Consider the variation of the TGD game with multigraphs. For all \(\iphase, k \in \N^+\) and \(\Tmax \ge 4\), there is an infinite family of temporal graphs \(\{\G_n\}_{n \in \N}\) such that the graph \(G_n\) has \(\Theta(n)\) nodes and the minimum number of rounds required to satisfy the TGD game grows in \(\Omega(n (T_\mathrm{max} - 3) / (\iphase k))\). Also, all graphs in the family are temporally connected.
\label{cor:multigraph-lower-bound}
\end{corollary}

On a more positive note, the \(\FDMain\) algorithm (\Cref{alg:follow-discover}) translates as well.

\begin{appendixtheoremrep}
Consider the variation of the TGD game with multigraphs.
\Cref{alg:follow-discover} wins any instance of the TGD game on a graph \(\G\) in at most \(6\edgesetsize{\G} + \decc{\G} \left\lceil \Tmax/\iphase \right \rceil\) rounds.
\end{appendixtheoremrep}

Note that here, we count every multiedge individually.
In essence, this means that we only pay additional rounds for the additional multiplicity of the edges.

\begin{appendixproof}
This result follows since the proofs of \Cref{thm:follow-discover}, \Cref{lem:follow-edge} and thus \Cref{cor:follow-whole-component} hold analogously on multigraphs.
\end{appendixproof}
\end{toappendix}

In summary, our \(\FDMain\) algorithm works if \(\ddv\) only gets infection-time feedback or if we allow multiedges, lifting the two most restrictive assumptions previously made. We also prove that variations where \(\ddv\) needs to ensure the non-existence of a high number of edges or labels (such as the unknown static graph or multilabel variations) do not allow for significant improvements over the trivial algorithm and that our \(\FDMain\) algorithm is not applicable. Results and proofs of this section have been moved to the appendix due to space constraints.


\section{Experimental Evaluation} \label{sec:experiments}
The gap between the lower and upper bounds for the TGD problem is small, but only tells us about the worst-case performance, motivating us to investigate the performance of our TGD algorithm on common synthetic graphs and real-world data.
We formulate three hypotheses we aim to test.

\begin{hypothesis}
The number of rounds required to discover a temporal graph is linear in the number of edge labels.
\end{hypothesis}
This first hypothesis is motivated by the worst-case analysis from \Cref{thm:follow-discover}.
We aim to test how closely the performance of our algorithm in practice matches this theoretical bound.
Note, this also takes into account the effect of the additional optimization described in the setup.

\begin{hypothesis}
Graphs with higher density spend fewer rounds on component discovery.
\label{hypo:component-discovery}
\end{hypothesis}
The \(\FDMain\) algorithm works in two distinct phases: (1) the main loop in \Cref{alg:follow-discover} discovers new \dECCs{} (the \emph{component discovery phase}) and (2) the \(\ex\) routine explores the found components (the \emph{component exploration phase}).
Both of these stages require \(\ddv\) to spend rounds, and their respective cost is dependent on the structure of the graph to be explored.
As we prove in \Cref{thm:follow-discover}, the component discovery phase is triggered at most \(\defaultdecc\) times.
We hypothesize that \(\decc{\G}\) tends to be lower in denser graphs as the components tend to merge as more edges are inserted, which would lead to a relatively lower cost for component discovery as compared to component exploration.

\begin{hypothesis}
In Erdős-Renyi graphs, the percentage of rounds spend on component discovery shows a threshold behavior in \(\Tmax / (\nodesetsize{\defaultGraphName} \cdot p)\). This is mediated by \(\defaultdecc\).
\label{hypo:threshold}
\end{hypothesis}
\Cref{hypo:threshold} \todo{something is broken here with the referencing} is a specification of \Cref{hypo:component-discovery} for Erdős-Renyi graphs.
We can view this hypothesis as the temporal extension of the typical threshold behavior that static Erdős-Renyi graphs exhibit in \(np\) \cite{erdos1960evolution}.
We also conjecture this behavior holds provably in \Cref{conj:erdos}. \todo{Maybe we want to go a bit more into detail into what exactly the threshold behavior is. I am not sure whether everyone is so used to thresholds or knows the Erdos Renyi results.}

\paragraph{Setup.} We perform our evaluation on two data sets.
One is synthetically generated, while the other is real-world data.
First, we evaluate our algorithm on Erdős-Renyi graphs. While the model has been developed for static graphs, it is commonly extended to temporal graphs \cite{Angel_Ferber_Sudakov_Tassion_2020}.
To generate a simple temporal graph from an Erdős-Renyi graph, we simply choose each edge label uniformly at random from the set \([\Tmax]\).
This means that \(\Tmax\) is now the third parameter to generate these graphs, in addition to the classical ones \(n\) (the number of nodes) and \(p\) (the density parameter).
Denote such a graph as \(ERT(n, p, \Tmax)\).\\
Secondly, to evaluate our algorithm on real-world data, we employ a data set from the Stanford Large Network Dataset Collection \cite{kumar2021deception}.
The \texttt{comm-f2f-Resistance} collection is described by the project as a set of 62 “dynamic face-to-face interaction network[s] between a group of participants”.

In our implementation, we employ a small optimization upon \Cref{alg:follow-discover}. We skip what we call \emph{redundant seed infections}.
A seed infection \((v, t) \in \nodeset{\G} \times [0, \Tmax]\) is \emph{redundant} if we already know the labels of all edges adjacent to \(v\) at the time in the algorithm we would perform this seed infection.
This does not depend on \(t\).
For an illustration, consider \Cref{fig:follow-execution}.
There, the \(\FDMain\) algorithm would, after discovering the only edge adjacent to the leftmost node by a seed from its other adjacent node, perform seed infections at the leftmost node even though we already know the label of the only edge we could possibly discover.

We run \(\FDMain\) on all graphs from the SNAP dataset and for a wide range of parameters for the temporal Erdős-Renyi graphs.
We test with 1 to 100 nodes in steps of size 5, for probabilities $p \in \{.01, .05, .1, .15, .2, .25,$ $.3, .35, .4, .5, .7, .9\}$.
We pick \(\Tmax\) as a factor of \(n\), testing with $\Tmax  / n \in \{.05, .1, .2, .3, .4, .5, .7, .9, 1, 2, 3, 5, 7, 10\}$.
Similarly, \(\iphase\) is 1 or a multiple of \(\Tmax\), namely \(\iphase / \Tmax \in \{.01, .05, .1, .3, .5\}\)  
We record the number of rounds needed to complete the discovery and how many of these rounds are spent in the component discovery versus the component exploration phases of the algorithm. Both our implementation and analysis code are available under a permissive open-source license and can be used to replicate our findings.\footnote{The repository is publicly available and will be linked after the blind review.}

\begin{figure*}[t]
    \begin{subfigure}{0.73\textwidth}
        \centering
        \includegraphics[width=\linewidth]{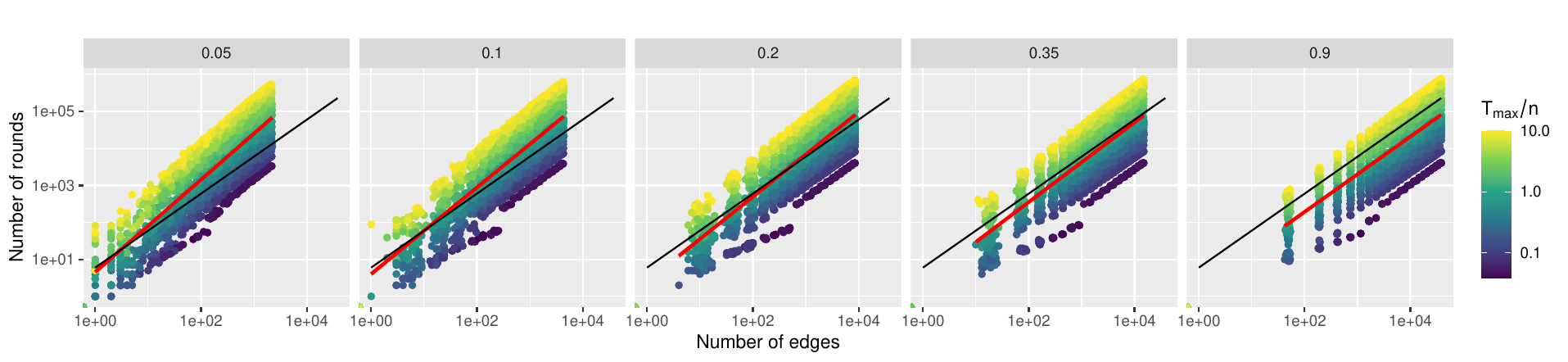}
        \caption{Erdős-Renyi graphs for different densities \(p\)}
        \label{fig:rounds_by_edge_count:erdos}
    \end{subfigure}%
    \hspace{0.04\textwidth}
    \begin{subfigure}{0.2\textwidth}
        \centering
        \includegraphics[width=\linewidth]{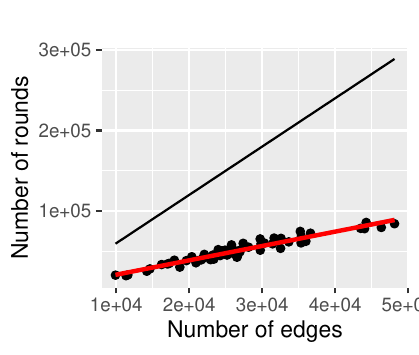}
        \caption{SNAP dataset}
        \label{fig:rounds_by_edge_count:snap}
    \end{subfigure}
    \caption{Number of rounds by number of edges for the \(\FDMainHack \) algorithm on the two data sets. The black line indicates the \(6\edgesetsize{\G}\) line (our theoretical bound for the component exploration phase), and the red line is the trend line (i.e.,~the linear regression).}
    \label{fig:rounds_by_edge_count}
\end{figure*}
\begin{figure*}[t]
    \begin{subfigure}{0.31\textwidth}
    \centering
    \includegraphics[width=\linewidth]{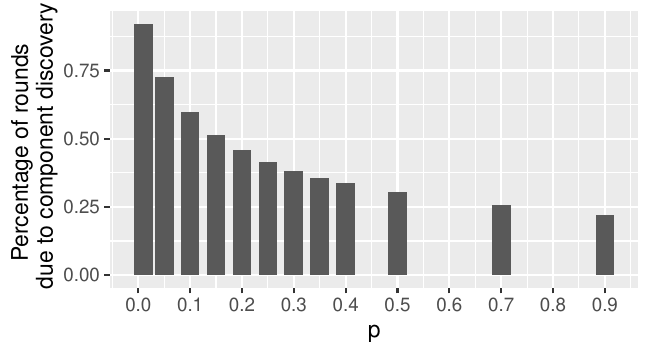}
    \caption{\label{fig:erdos-renyi_component-discovery-percentage-by-p}The percentage of rounds the \(\FDMainHack\) algorithm spends in the component discovery phase on temporal Erdős-Renyi graphs by density \(p\).}
    \end{subfigure}%
    \hfill
    \begin{subfigure}{0.31\textwidth}
        \centering
        \includegraphics[width=\linewidth]{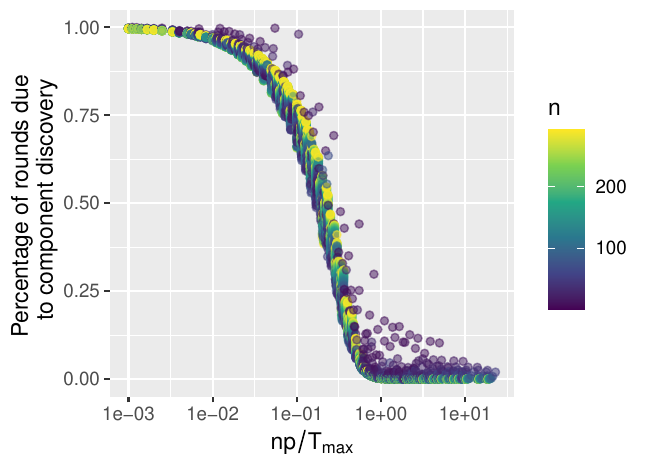}
        \caption{The percentage of rounds spent in the component discovery phase by \(np / \Tmax\).
        \label{fig:erdos-renyi_component-discovery-percentage-by-n-times-p-over-Tmax}}
    \end{subfigure}%
    \hfill
    \begin{subfigure}{0.31\textwidth}
        \centering
        \includegraphics[width=\linewidth]{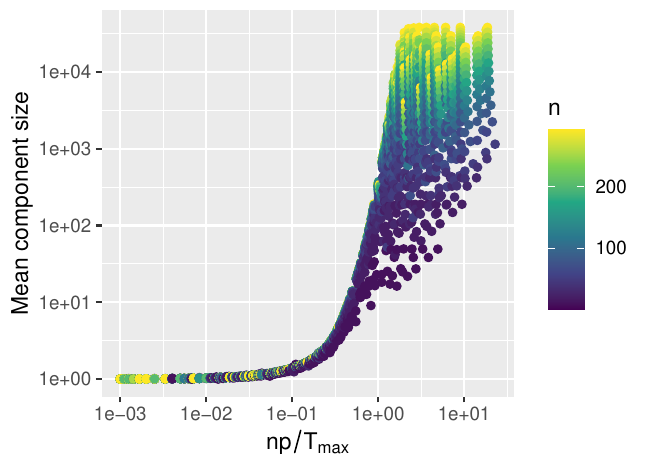}
        \caption{Mean size of the \dECCs{} by \(np / \Tmax\).
        \label{fig:erdos-renyi_component_mean_size-by-n-times-p-over-Tmax}}
    \end{subfigure}
    \caption{Cost of component discovery in the \(\FDMainHack \) algorithm. 
    Subfigures (b) and (c) use logarithmic axes, and the color denotes the number of nodes.}
    \label{fig:threshold}
\end{figure*}
\subsection{Results}
We now critically evaluate our hypotheses against the data thus obtained and compare the effects between the different data sets and parameters.
Particularly, we pay attention to evidence on how the different hypotheses interact.
Finally, we give \Cref{conj:erdos} as a result of our analysis of \Cref{hypo:threshold}.

\paragraph{Hypothesis 1.}
\begin{toappendix}
    With \Cref{fig:rounds_by_edge_count:erdos:full}, we provide an extended version of \Cref{fig:rounds_by_edge_count:erdos} for all tested values.
    Some of the values were ommited from the main part due to space constraints.

    \begin{figure}[tbhp]
    \centering
        \includegraphics[width=0.73\textwidth]{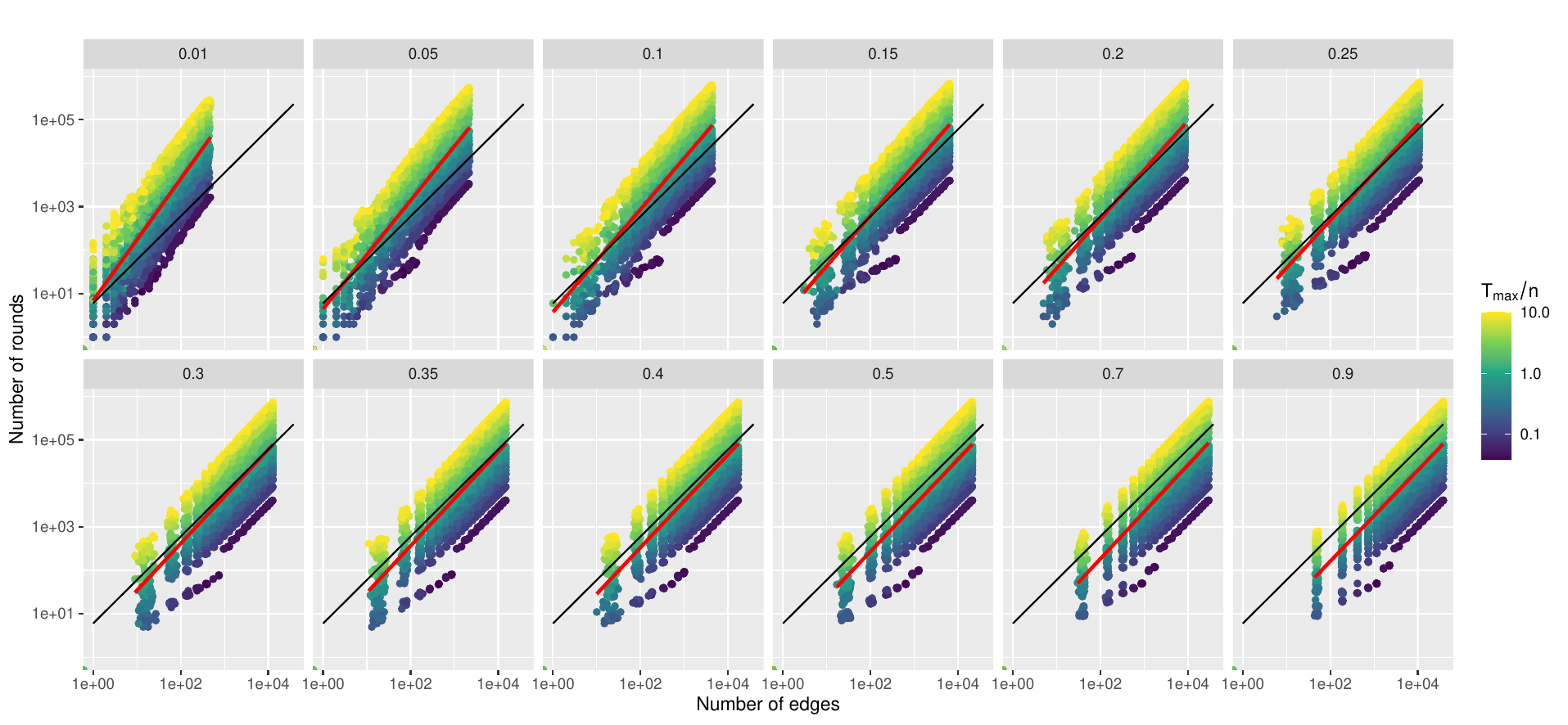}
        \caption{Erdős-Renyi graphs for different densities \(p\)}
        \label{fig:rounds_by_edge_count:erdos:full}
    \end{figure}
\end{toappendix}

To investigate our first hypothesis, we compare how the number of rounds relates to the number of edges.
See \cref{fig:rounds_by_edge_count} for the results on our datasets.
We can see that in Erdős-Renyi graphs, the relationship closely follows the \(6\edgesetsize{\G}\) line predicted by our theoretical results.
Clearly, this effect is influenced by other parameters such as \(p\), \(\Tmax\), and \(n\), whose roles we will examine in the discussion of the other hypotheses.
However, if we consider graphs with the same \(p\) and ratio \(\Tmax / n\) (i.e.,~one facet and one color in \cref{fig:rounds_by_edge_count}), we see that the relationship is strictly linear—the points form a tightly distributed straight line.
 We see that the trend for \(p \le 0.25\) is slightly above \(6\edgesetsize{\G}\) and slightly under it for larger values of \(p\).
This occurs since, when \(p\) is small, more time is spent on component discovery than on component exploration.
We will explore this more thoroughly in the discussion of the results regarding \Cref{hypo:component-discovery},

In the SNAP data set, the trend is linear in the number of edges, but significantly less than \(6\edgesetsize{\G}\) rounds are required to discover the graph.
In fact, the gradient of the regression line is only 1.78.
This can be explained by the optimization skipping redundant infections as outlined in our setup, as this enables the algorithm to require less than the \(6\) infections per edge, which would otherwise be strictly required.
In summary, while there is a strong linear relationship following the \(6\edgesetsize{\G}\) line, the number of rounds also significantly depends on other properties of the graph.
We can see these effects in both synthetic and real-world data.

\paragraph{Hypothesis 2.}
\label{sec:org8ac9a1e}

As this hypothesis is about the relationship between graph density and time spent on component discovery, we only analyze it on the Erdős-Renyi graphs. The SNAP dataset does not have significant differences in graph density.

In \Cref{fig:erdos-renyi_component-discovery-percentage-by-p}, we plot the percentage of time spent on component discovery dependent on the parameter \(p\) (which specifies the density of the graph).
We observe a clear and strong, inversely proportional relationship.
This leads us to accept \cref{hypo:component-discovery}.
This is explained from the theoretical analysis of \(\FDMain\), as we expect denser graphs to have fewer \dECCs{}, thus less need to discover new components.

\paragraph{Hypothesis 3.}
\label{sec:org6b7b535}

Examining \cref{fig:erdos-renyi_component-discovery-percentage-by-n-times-p-over-Tmax} \todo{where is the figure gone}, we see a threshold between \(np/\Tmax = 0.01\) and \(np/\Tmax = 1.00\) where we move from spending only a small amount of rounds on component discovery to spending close to all of our rounds on component discovery.
Our hypothesis, that this is mediated by the size of the \dECCs{}, is supported by the results in \cref{fig:erdos-renyi_component_mean_size-by-n-times-p-over-Tmax}.
The additional differentiation by color for larger values of \(np / \Tmax\) can be explained since the size of a component is constrained by the size of the graph.

Note that big \dECCs{} imply that we can discover many edges in the component exploration phase for a single node explored in the component discovery phase.
That means the plotted percentage shrinks as the average size of \dECCs{} increases.
This behavior is similar to that seen in static Erdős-Renyi graphs, where if \(np=1\), there almost surely is a largest connected component (in the classical sense) of order \(n^{2/3}\).
And if \(np\) approaches a constant larger than 1, there asymptotically almost surely is a so-called giant component containing linearly many nodes \cite{erdos1960evolution}.
This motivates us to conjecture the following similar behavior.
\begin{conjecture}
Let \(\iphase = 1, p \in [0,1], \Tmax, n \in \N\) and \(\G \coloneqq ERT(n, p, \Tmax)\). Then
\begin{itemize}[topsep=5pt]
\item if \(\Tmax/(np) \xrightarrow{n \to \infty} c\) where \(c < 1\), then \(\decc{\G} / \edgesetsize{\G} \xrightarrow{n \to \infty} 0\), and
\item if \(\Tmax/(np) > 1\), then almost surely, all \dECCs{} have size at most \(\BigO(\log n^2) = \BigO(\log n)\).
\end{itemize}
 \ifpaper \else {\vspace{-1em}}\fi
\label{conj:erdos}
\end{conjecture}
Investigating the connectivity behavior of random temporal graphs in a way that respects their inherent temporal aspects is an interesting and little understood problem \cite{casteigts_threshold,casteigts2024simple}.
The authors theoretically study sharp thresholds for connectivity in temporal Erdős-Renyi graphs, and conjecture about a threshold for the emergence of a node-based giant (i.e.,~linear size) connected component.
This conjecture can be seen as capturing the equivalent behavior under waiting time constraints (i.e.,~where edges are only considered if their edge labels do not differ too much).

\section{Conclusion} \label{sec:conclusion}
We give a comprehensive theoretical and empirical study of TGD in temporal graphs. 
\(\FDMain\) provides an efficient TGD strategy requiring only \(6\edgesetsize{\G} + \decc{\G} \left\lceil \Tmax/\iphase \right \rceil\) rounds.
Our lower bound proves that any algorithm must spend at least \(\Omega(\edgesetsize{\G})\) rounds, showing \(\FDMain\) is close to optimal.
Our empirical analysis highlights the relevance of our theoretical results for practical applications and gives rise to interesting insights of its own.
We see that on Erdős-Renyi graphs, the observed performance of \(\FDMain\) matched our theoretical analysis.
On real-world data from the SNAP collection, the algorithm even slightly outperforms our predictions.
Finally, we observe a close link between the parameters of the temporal Erdős-Renyi model, the temporal connectivity structure of the resulting graphs, and our algorithmic performance, creating a bridge between our theoretical insights on \(\iphase\)-connected components and their empirical behavior.

Future work can tighten the lower bound from \Cref{thm:temporally-connected-graph-discovery-lower-bound} and give a lower bound that is tight in \(k\) and \(\iphase\),
Another avenue is to investigate \(\defaultdecc\) further. In particular, to prove or disprove the observed threshold in Erdős-Renyi graphs (\Cref{conj:erdos}) and investigate an analog to the static \(\ln n / n\) threshold. 
Finally, future research can explore more variations of TGD, such as finding specific nodes or checking for structural properties instead of discovering the whole graph.


\newpage

\bibliographystyle{named}
\bibliography{papers}

\listoftodos{}

\clearpage
\onecolumn
\appendix
\renewcommand{\appendixpagename}{Supplementary Material for Submission XXX}

\end{document}